\documentclass[graybox]{svmult}

\usepackage{mathptmx}       
\usepackage{helvet}         
\usepackage{courier}        
\usepackage{makeidx}         
\usepackage{multicol}        
\usepackage[bottom]{footmisc}

\makeindex             

\usepackage{amsmath,amssymb,xspace,ctable,url,citesort,multirow,wrapfig,txfonts,fge}
\usepackage{pstricks,pst-node,pst-coil,pst-text,multido,pst-3d,pst-tree,pst-grad}
\usepackage{nicefrac} 
\usepackage{enumerate}
\usepackage{longtable}

\usepackage{frcursive}

\newcommand{\FI}[1]{Fig.~\ref{#1}\xspace}

\newcommand{\IE}{{\em i.e.}\xspace}

\newcommand{\comment}[1]{}

\begin{document}

\title*{Densely Entangled Financial Systems}

\author{Bhaskar DasGupta and Lakshmi Kaligounder}

\institute{Bhaskar DasGupta \at Department of Computer Science, University of Illinois at Chicago, Chicago, IL 60607-7053, \email{bdasgup@uic.edu}
\and
Lakshmi Kaligounder \at Department of Computer Science, University of Illinois at Chicago, Chicago, IL 60607-7053, \email{lkalig2@uic.edu}}

\maketitle

\abstract{
In~\cite{Z11} Zawadoski introduces a banking network model in which the asset and counter-party risks are treated separately and the banks 
hedge their assets risks by appropriate OTC contracts. In his model, each bank has only two counter-party neighbors, 
a bank fails due to the counter-party risk only if at least 
one of its two neighbors default, and such a counter-party risk is a low probability event. 
Informally, the author shows that the banks will hedge their asset risks by appropriate OTC contracts, and, 
though it may be socially optimal to insure against counter-party risk, in equilibrium banks will {\em not} choose to insure this low probability event. 
In this paper, we consider the above model for more general network topologies, namely when each node has exactly $2r$ counter-party neighbors for some integer $r>0$.
We extend the analysis of~\cite{Z11} to show that as the number of counter-party neighbors increase the probability of counter-party risk also increases, and in particular 
the socially optimal solution becomes privately sustainable when each bank hedges its risk to at least $\nicefrac{n}{2}$ banks, 
where $n$ is the number of banks in the network, \IE, when $2r$ is at least $\nicefrac{n}{2}$, banks not only hedge their asset risk but also hedge its counter-party risk.
}

\section{Introduction}

Economic stability has received special attention during the past several years mainly because of the economic downturn experienced globally in the recent past. 
This attention has generated renewed interest in evaluating how important the economic stability is as having an unstable economy can pave way to an 
economic crisis each time when the global markets sees a downward trend.
Financial instability and its effects on the economy can be very costly due to its contagion or spillover effects to other parts of the economy and it is 
fundamental to have a sound, stable and healthy financial system to support the efficient allocation of resources and distribution of risks across the economy. 
Financial risk management is a critical component of maintaining economic stability. Hedging is a risk management option that protects the owner of an asset from loss. 
It is the process of shifting risks to the futures market. The risks in the market must be first identified in order to manage the risk. 
To identify the risk one must examine both the immediate risk (asset risk) as well as the risk due to indirect effects (counter-party risk).
Though hedging will minimize overall profit when markets are moving positive, it also helps in reducing risk during undesirable market conditions. 
However, as the owner hedges his/her asset risk to protect  against defaults, the owner also gets exposed to the counter-party risk. 

In~\cite{Z11} Zawadoski introduces a banking network model in which the asset and counter-party risks are treated separately, and 
showed that, under certain situations, banks do not manage counter-party risk in equilibrium. In his model, each bank has only two counter-party neighbors, 
a bank fails due to the counter-party risk only if {\em at least} one of its two neighbors default, and 
such a counter-party risk is an event with {\em low} probability. Informally, the author shows that the banks will hedge their asset risks by 
appropriate OTC contracts, and, though it may be socially optimal for banks to insure against counter-party risk, in equilibrium banks 
will {\em not} choose to insure this low probability event. The OTC contract not only creates a contagion but also creates externalities 
which undermines the incentives of the banks to avert contagion. The model uses short term debt to finance their real asset. 
The failure in this model is from the liability side, where the investors run on the banks when they do not trust the bank, \IE, 
the investors do not roll over the debts of the banks. Hence the contagion can be avoided only by increasing the equity and {\em not} by providing liquidity.

In this paper, we consider the above model for {\em more general} network topologies, namely when each node has exactly $2r$ counter-party neighbors for some integer $r>0$. 
We extend the analysis of~\cite{Z11} to show that as the number of counter-party neighbors increase the probability of counter-party risk also 
increases\footnote{Thus, the owners will decide to hedge their counter-party risk thereby helping to contain an economic crisis.}, and in particular, 
the socially optimal solution becomes privately sustainable when each bank hedges its risk to a sufficiently large number of other banks. 
The counter-party risk can be hedged by holding more equity, buying default insurance on their counter-parties or collateralizing OTC contracts. 
Since holding excess capital or collateralizing OTC contracts is a wasteful use of scarce capital~\cite{Z11}, when the banks choose to hedge their counter-party risk they buy the 
default insurance on their counter-parties. More precisely, our conclusions for the general case of $2r$ neighbors are as follows:
\begin{itemize}
\item
All the banks will still decide to hedge their asset risks.

\item
If the number of counter-party neighbors is at least $\nicefrac{n}{2}$, then all banks will decide to insure their counter-parties, and socially optimal solution in case of two counter-parties for 
each bank now becomes privately optimal solution.

\item
In the limit when the number of banks $n$ in the network tend to $\infty$, as the number of counter-party neighbors approach $n-1$, 
failure of {\em very few} of its counter-party banks will not affect a bank. 
\end{itemize}

\section{Related Prior Research Works}

As we have already mentioned, 
Zawadowski~\cite{Z11} introduced a banking model in which asset risk and counter-party risk are treated separately, showed that banks always prefer to hedge their asset 
risk using OTC contracts and also showed that banks do not hedge their counter-party risk even though hedging counter-party risk is possible and socially desirable. 
Allen and Gale~\cite{AG00} showed that interbank deposits help banks share liquidity risk but expose them to asset losses if their counter-party defaults. 
Their model cannot be used to understand the contractual choices in case of OTC derivatives as they modeled the liquidity risk.  
Babus~\cite{B07} proposed a model in which links are formed between banks which serves as an insurance mechanism to reduce the risk of contagion. 
Allen and Babus~\cite{AB09} pointed out that graph-theoretic concepts provide a conceptual framework used to describe and analyze the banking network, and showed that 
more interbank links provide banks with a form of coinsurance against uncertain liquidity flows. Gai and Kapadi~\cite{GK10} showed that more interbank links increase the 
opportunity for spreading failures to other banks during crisis. Several prior researchers such as~\cite{AB09,E04,KWG09,NYYA07} commented that 
graph-theoretic frameworks may provide a powerful tool for analyzing stability of banking and other financial networks. Kleindorfer {\em et al}.~\cite{KWG09} 
argued that network analyses can play a crucial role in understanding many important phenomena in finance. Freixas {\em et al}.~\cite{FPR00} explored the case of banks that 
face liquidity fluctuations due to the uncertainty about consumers withdrawing funds. Iazzetta and Manna~\cite{H09} analysed the monthly data on deposit exchange 
to understand the spread of liquidity crisis using network topology. 

Babus~\cite{B06} studied how the trade-off between the benefits and the costs of being linked changes depending on the network
structure and observed that, when the network is maximal, liquidity can be redistributed in the system to make the risk of contagion minimal. 
Corbo and Demange~\cite{CD10} explored the relationship of the structure of interbank connections to the contagion risk of defaults
given the exogenous default of set of banks.
Nier {\em et al}.~\cite{NYYA07} explored the dependency of systemic risks on the structure of the banking system via network theoretic approach and 
the resilience of such a system to contagious defaults. Haldane~\cite{H09} suggested that contagion should be measured based on the interconnectedness of each institution within 
the financial system. Liedorp {\em et al}.~\cite{LMKK05} investigated if interconnectedness in the interbank market is a channel through which banks affect each others riskiness, 
and argued that both large lending and borrowing shares in interbank markets increase the riskiness of banks active in the dutch banking market. 
Kiyotaki and Moore~\cite{KNM97} studied the chain reaction caused by the shock in one firm and the financial difficulties in other firms due to this chain reaction. 
Acharya and Bisin~\cite{AVB09} compared centralized markets to OTC markets and showed that counter-party risk externalities can lead to excessive default and 
production of aggregate risk. Caballero and Simsek~\cite{CRS09} concluded that OTC derivatives are not the sole reason for the inefficiency of financial networks. 
Pirrong~\cite{PC09} argued that central counter-parties (CCP) can also increase the systemic risk under certain circumstances and hence the introduction of CCP will not guarantee to 
mitigate the systemic risk. 
Zawadowski~\cite{Z10} showed that complicated interwoven financial intermediation can be a reason for inefficient financial networks, and hence OTC are not the only reason for financial 
instability. Stulz~\cite{S09} showed that exchange trading has both benefits and costs compared to OTC trading, and argued that 
credit default swaps (CDS) did not cause the crisis since they worked well during much of the first year of the crisis.
Zhu and Pykhtin~\cite{SM07} showed that modeling credit exposures is vital for risk management application, while modelling credit value adjustment (CVA) is necessary step for 
pricing and hedging counter-party credit risk. Corbo and Demange~\cite{CD10} showed that introduction of central clearing house for credit default swaps will mitigate the counter-party risk. 
Gandhi {\em et al}.~\cite{GLA11} paralleled and complemented the conclusion of~\cite{CD10}, \IE, the creation of central clearing house for CDS contracts may not reduce the counter-party risk.

\section{The basic model}

The model has $n>3$ banks and three time periods $t=0,1,2$ termed as {\em initial}, {\em interim} and {\em final}, respectively. 
Each bank has exactly $2r$ counter-party neighbors for some integer $r>0$ (see \FI{fig1} for an illustration). 
The unit investment of each bank in the {\em long term} real asset yields a return of 
$\displaystyle R + \!\!\! \sum_{k\,=\,|\,i-1\,|}^{|\,i-r\,|}\!\!\! \varepsilon_k \,\,- \!\!\! \!\!\! \sum_{k\,=\,|\,i-(r+1)\,|}^{|\,i-2r\,|} \!\!\!\!\!\!\!\!\!   \varepsilon_k$ at $t=2$, where 
$R=\left\{
\begin{array}{ll} 
R_H, & \mbox{if the project succeeds} 
\\
R_L<R_H, & \mbox{if the project fails}
\end{array}
\right.
$, 
and each $\varepsilon_k$ is realized at $t=2$ taking values of $u$ or $-u$ each with probability $\nicefrac{1}{2}$. 
For each unit investment made by the bank at $t=0$, the investor lends $D\geq 0$ as short term debt and equity $1-D\geq 0$ is the bank's share. 
The short term debt has to be rolled over at time period $t=1$ for the banks to operate successfully. 
Thus the debt holders have an option to withdraw funding and force the bank to liquidate the real project.

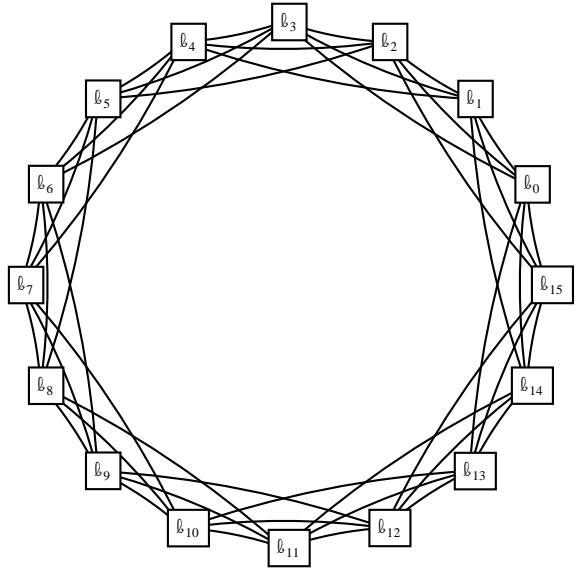
\begin{figure}[htbp]
\begin{pspicture}(-6,-4)(4,4)
\newcounter{CtA}
\newcounter{CtB}
\newcounter{Jj}
\newcommand{\Wheel}[3]{{
\color{black}
\psset{unit=3.5}
\SpecialCoor
\degrees[#1]
\multido{\ia=1+1}{#1}{
  \setcounter{CtA}{\ia} 
  \stepcounter{CtA}
  \setcounter{CtB}{#2} 
  \multido{\ib=\value{CtA}+1}{\value{CtB}}
    {#3(1;\ia)(1;\ib)}
}
\multido{\i=1+1}{#1}{
  \rput(1;\i){
  \setcounter{Jj}{\i} 
  \addtocounter{Jj}{-1} 
  \psframebox[fillstyle=solid,fillcolor=white]{\tiny\bf\textcursive{b}$_{\,\arabic{Jj}}$}
  }
}
}}
\psset{xunit=1cm,yunit=1cm}
{
\Wheel{16}{3}{\pcarc[arcsepA=0pt,arcsepB=0pt,arcsep=0pt]}
}
\end{pspicture}
\caption{\label{fig1}An illustration of a network of $16$ banks \textcursive{b}$_0$, \textcursive{b}$_1$, $\dots$, \textcursive{b}$_{15}$ with $r=3$.
Each \textcursive{b}$_i$ is connected to the banks \textcursive{b}$_{\,i-3 \pmod {16}}$, \textcursive{b}$_{\,i-2 \pmod {16} }$, \textcursive{b}$_{\,i-1 \pmod {16} }$, 
\textcursive{b}$_{\,i+1 \pmod {16} }$, \textcursive{b}$_{\,i+2 \pmod {16} }$ and \textcursive{b}$_{\,i+3 \pmod {16} }$.}
\end{figure}

Let $e\in\{0,1\}$ be the unobservable effort choice such that a bank needs to exert an effort of $e=1$ at both time period $t=0$ and $t=1$ for the project to 
be successful (\IE, $R=R_H$). At $t=1$ the project can be in one of the two states: a ``bad'' state with probability $p$ or a ``good'' state with probability $1-p$, 
irrespective of the effort exerted by the bank. At a ``bad'' state the project of one of the randomly chosen bank fails and delivers $R_L$, even if $e=1$ at both time periods $t=0$ and $t=1$.  
Unless the bank demand collateral from its counter-parties, if the bank defaults at $t=1$ then all the hedging liabilities of the defaulted bank gets cancelled, 
the investors liquidate the bank and take equal share of $L$ (the value of the bank when it is liquidated). 
If the bank survives till $t=2$ and the counter-party risk gets realized then, the bank has to settle the counter-party hedging contract {\em before} paying its debt. 

We use the following notations for four specific values of the probability of bad state $p$:
\begin{description}
\item[${p^{\mathrm{soc}}}$:]
if $p<p^{\mathrm{soc}}$, then irrespective of the number of counter-party neighbors there is no need for counter-party insurance even from social perspective. 

\item[${p^{\mathrm{ind}}}$:]
if $p>p^{\mathrm{ind}}$ then the banks will not buy counter-party insurance as the private benefits of insuring exceeds the cost. 

\item[${p^{\mathrm{term}}}$:]
if $p<p^{\mathrm{term}}$ then the banks will continue to prefer short term debt.

\item[${p^{\mathrm{aut}}}$:]
if $p<p^{\mathrm{aut}}$ then no bank will have an incentive to hold more equity and borrow less.
\end{description}

\subsection*{Parameter Restrictions and Assumptions}

The following parameter restriction are adopted from~\cite{Z11} to make them consistent for a network model with $2r$ counter-parties.
$B$ is the banks' private benefit with the subscript representing the specific time period and $X$ denotes the additional non-pledgable payoff. 
Inequality~\eqref{eq1} ensures that the investors will choose to roll over the debt at $t=1$ when the project is expected to succeed (\IE, $R=R_H$), and the investors will decide to 
liquidate the bank at $t=1$ if the bank's project is expected to fail (\IE, $R=R_L$). Inequality~\eqref{eq2} implies that it is socially optimal to exert effort. Inequality~\eqref{eq3} 
ensures that banks have to keep positive equity to overcome moral hazard. Inequality~\eqref{eq4} ensures that,  counter-party risk of the bank is large enough to lead to contagion
but small enough that the bank does not want to engage in risk-shifting.
\begin{equation}
R_L< L< \left(1 - \dfrac{p}{n} \right) \big(R_H+ X\big) - B_1
\label{eq1}
\end{equation}

\begin{equation}
R_H - R_L>\dfrac{2 B_1}{1-\frac{p}{n}}
\label{eq2}
\end{equation}

\begin{equation}
 B_1 \geq R_H - 1 + X
\label{eq3}
\end{equation}

\begin{equation}
B_1-X< u < \dfrac{ B_1 \left(2^{2r}-\displaystyle \sum_{K=0}^r \frac{2r!}{K! \, \big(2r-K\big)!} \right)}{\displaystyle 2r\sum_{K=0}^{r}\frac{2r-2K}{K! \, \big(2r-K\big)!}}
\label{eq4}
\end{equation}

\begin{equation}
\beta >\nicefrac{1}{2} \,\,\,\,\text{ and }\,\,\,\, \beta > \dfrac{1-\big(1-p\big)\, \big(R_H+X- B_1\big)-p\,L}{\big(1-p\big) B_1}
\label{eq5}
\end{equation}

\begin{equation}
2u \leq B_0 < \big(1-p\big) B_1
\label{eq6}
\end{equation}

\section{Our Results}

Our results imply that when the number of counter-party neighbors is at least $\nicefrac{n}{2}$, the socially optimal outcome become privately sustainable. 

\begin{theorem}\label{main-thm}
If the probability of bad state is $p\in[\left.0,p^\star\right)$, where $p^\star = \min \left\{ p^{\mathrm{ind}}, p^{\mathrm{aut}}, p^{\mathrm{term}} \right\}$ then the followings hold.
\begin{enumerate}[(a)]
\item  
Banks endogenously enter into OTC contracts as shown by Zawadowski in~{\em\cite{Z11}}.

\item 
Banks borrow $D<1$ for short term at $t=0$ at an interest rate $R_{\ensuremath{\imath},0}>1$ and $R_{\ensuremath{\imath},1}=1$ as shown by Zawadowski in~{\em\cite{Z11}}.

\item 
In a bad state, failure of a single bank leads to run on all banks in the system  only when $2r<\nicefrac{n}{2}$.

\item
If a bank loses at least $r$ counter-parties, it needs a debt reduction of $I = ru - B_1 + X > 0$.

\item 
The contagious equilibrium stated by Zawadowski in~{\em\cite{Z11}} exist only if $2r<\nicefrac{n}{2}$. When $2r\geq\nicefrac{n}{2}$ banks insure against counter-party risk using default insurance.

\item 
If $p\in \left(p^{\mathrm{soc}}, p^{\mathrm{ind}}\right)$ then the socially optimal outcome is sustainable in equilibrium.
\end{enumerate}
\end{theorem}

Theorem~\ref{main-thm}(a) is proved in Lemma~\ref{claim6}, whereas Theorem~\ref{main-thm}(c),(e),(f) are showed in Lemma~\ref{claim13}. 
Theorem ~\ref{main-thm}(d) follows from the derivations of parameter values as described in Section~\ref{para-calc}; 
these derivations follow from the work in~\cite{Z11} and are provided in the appendix.
Theorem ~\ref{main-thm}(b) uses the same proof as that in~\cite{Z11} and is therefore omitted.

\subsection{Proofs of Theorem~{\em\ref{main-thm}}}
\label{para-calc}

The derivations of the parameters and their values described in items (I) -- (VIII) below follow from the work in~\cite{Z11} since they are not affected by changing the number of 
counter-party neighbors from $2$ to $2r$.  For the sake of completeness, these derivations are provided in the appendix.
\begin{enumerate}[(I)]
\item 
The maximum amount of borrowing at $t=0$ that can be roll over at $t=1$ is given by: 
\[
D^{\max}\big(R_{\ensuremath{\imath},0}\big)= \dfrac{R_H- B_1+X}{R_{\ensuremath{\imath},0}} < 1
\]
and the expected payoff of the above bank borrowing $D^{\max}$ is $B_1$.

\item 
In a stable system where all banks buy counter-party insurance, the price per unit of default insurance is
\[
s^{\mathrm{safe}} = \dfrac{1-\beta}{n}- \beta\,\dfrac{p}{n}
\]
(where the superscript ``safe'' denotes the insured system).

\item 
In a stable system where all banks buy counter-party insurance, the amount borrowed by the bank for unit investment in the project is
\[
D^{\mathrm{safe}} = \left(1-\dfrac{p}{n}\right)\,\big( R_H- B_1+X\big) + \dfrac{p}{n}\,L
\]

\item 
In a contagious system where banks decide not to buy insurance, the amount borrowed by the bank for unit investment in the project is 
\[
D^{\,\star} = \big(1-p\big)\, \big( R_H- B_1+X \big) +p\,L
\]

\item 
In a contagious system where banks decide not to buy insurance, the interest rate for the amount borrowed is  
\[
R^{\,\star}_{\ensuremath{\imath},0}= \dfrac{1}{1-p \left( \displaystyle 1-\frac{L}{R_H- B_1+X} \right)}
\]

\item 
$p^{\mathrm{term}} = \big(1 - \beta\big)\,\left(\dfrac {n}{n-1} \right) \,\left( \dfrac {B_0}{ R_H + X - L - B_1 \big(1-\beta \big) +\dfrac{R_L+ \beta B_1 - L}{n-1}}\right)$.

\item 
$p^{\mathrm{f.aut}} = \dfrac{ \big(1- \beta \big) \, \big(R_H + X - B_1 - R_L \big)}{R_H + X - L -B_1\,\big(1 - \beta\big) - \dfrac{\beta}{n}\big(R_H - R_L\big)}$.

\item 
$p^{\mathrm{aut}}=\min \left\{ p^{\mathrm{s.aut}},p^{\mathrm{r.aut}},p^{\mathrm{f.aut}}\right\}$, 
where the superscripts $\scriptstyle\mathrm{s.aut}$, $\scriptstyle\mathrm{r.aut}$ and $\scriptstyle\mathrm{f.aut}$ stand for safe autarky, risky autarky and full autarky, respectively.
\end{enumerate}

\begin{lemma}\label{claim7}
 $p^{\mathrm{ind}} = \dfrac{I \big(1-\beta\big)}{\beta  \, B_1}$
\end{lemma}

\noindent
({\em Informally, Lemma~\ref{claim7} states that the probability of bad state must be at least $p^{\mathrm{ind}} = \frac{I \big(1-\beta\big)}{\beta  \, B_1}$
for a bank's benefit to outweigh the cost of its counter-party insurance}).

\begin{proof}
\smartqed
The private incentive of a single bank to stay insured instead of deviating from counter-party insurance is when the following holds
(in the left-hand side of the inequality, the first term  is the payoff if no bank defaults, the second term is the payoff if any of the counter-parties of the bank default, 
the third term is money invested to buy the insurance on the bank's counter-parties and the fourth term is the bank's equity):
\[
\begin{array}{ll}
& 
\text{benefits from all banks being insured } \geq \text{ benefits from all but one bank not being insured}
\\
[0.1in]
\equiv
& 
\beta \left(1-\dfrac{\big(1+2r\big)p}{n} \right)\,B_1 + \beta \, \dfrac{2rp}{n}\,\big(B_1+I\big) - s^{\mathrm{safe}} \big(2rI\big) - \left(1- D^{\mathrm{safe}} \right) 
\\
&
\hspace*{0.3in}
\geq \beta \, \left(1-\dfrac{(1+2r).p}{n} \right)\,B_1- \left(1- D^{\mathrm{safe}} \right)
\\
[0.1in]
\equiv
& 
\beta \, \dfrac{2rp}{n} \big( B_1+I \big) - s^{\mathrm{safe}} \big( 2rI \big) \geq 0
\,\equiv\,
\beta \, \dfrac{2rp}{n} \big(B_1+I\big) -\left(\dfrac{1-\beta}{n}+ \beta \,\dfrac{p}{n} \right) \big(2rI\big) \geq 0
\\
[0.1in]
\equiv
& 
\dfrac{2\beta rp B_1}{n}+ \dfrac{2\beta rpI}{n}- \dfrac{2rI}{n} + \dfrac{2r\beta I}{n}- \dfrac{2\beta rpI}{n} \geq 0
\,\equiv\,
\beta r p  B_1 - rI +r \beta I \geq 0
\\
[0.1in]
\equiv
& 
\beta r p B_1 \geq r I - r\beta I 
\,\equiv\,
\beta p B_1 \geq I - \beta I
\,\equiv\,
p \geq \dfrac{I \big(1-\beta \big)}{\beta  B_1}
\,\Longrightarrow\,
p^{\mathrm{ind}} = \dfrac{I \big(1-\beta \big)}{\beta  B_1}
\end{array}
\]
\qed
\end{proof}

\begin{lemma}\label{claim8}
$p^{\mathrm{soc}} = \dfrac{2Ir\, \big(1-\beta \big)}{\big(n-1\big) \, \Big(R_H+X-L- B_1 \big(1-\beta \big) \, \Big)}$.
\end{lemma}

\noindent
({\em Informally, Lemma~\ref{claim8} states that the probability of bad state must be at most 
$p^{\mathrm{soc}} = \frac{2Ir\, \left(1-\beta \right)}{\left(n-1\right) \, \left(R_H+X-L- B_1 \left(1-\beta \right) \, \right)}$
for the social benefits to outweigh the social cost of its counter-party insurance}).

\begin{proof}
\smartqed
The social benefits outweigh the social cost of the counter-party insurance of the system if and only if the following holds:

\begin{longtable}{ll}
 & benefits from all banks being insured $\,\,\,\geq\,\,\,$ benefits from no bank being insured
\\[0.1in]
$\equiv$ & 
$\beta \, \left(1-\dfrac{\big(1+2r\big)p}{n} B_1+\beta \,\dfrac{2rp}{n} B_1+I \right)-s^{\mathrm{safe}} \big(2rI \big)-\left(1-D^{\mathrm{safe}} \right) 
     \,\,\, \geq \,\,\, \beta \big(1-p \big) B_1- \left(1 - D^{\,\star} \right)$
\\[0.1in]
$\equiv$ & 
$\beta B_1-\dfrac{\beta B_1p}{n}-\dfrac{2\beta  B_1p\,r}{n}+\dfrac{2\beta  B_1 p\,r}{n}+\dfrac{2\beta Ip\,r}{n}-\left(\dfrac{1-\beta}{n}+\dfrac{\beta \, p}{n} \right) \big(2rI\big)-1
     +\left(1-\dfrac{p}{n} \right) \, \big(R_H+X - B_1\big)+\dfrac{p\,L}{n}$
\\[0.1in]
& \hspace*{0.5in} $\,\,\, \geq \,\,\, \beta \, B_1 \big(1-p \big)-1 + \big(1-p \big) \, \left(R_H+X - B_1\right) +p\,L$
\\[0.1in]
$\equiv$ & 
$\beta B_1-\dfrac{\beta  B_1p}{n}+\dfrac{2\beta Ip\,r}{n}-\dfrac{2rI}{n}+\dfrac{2\beta rI}{n}-\dfrac{2\beta Ip\,r}{n}-1+R_H+X- B_1-\frac{pR_H}{n}-\dfrac{pX}{n}+\dfrac{p B_1}{n}+\dfrac{p\,L}{n}$
\\
& \hspace*{0.5in} $\,\,\, \geq \,\,\, \beta B_1-p\beta B_1-1 +R_H+X - B_1-pR_H-pX +p B_1 +p \,L$
\\[0.1in]
$\equiv$ & 
$-\dfrac{\beta B_1p}{n}-\dfrac{2rI}{n}+\dfrac{2\beta\, r\,I}{n}-\dfrac{p\,R_H}{n}-\dfrac{p\,X}{n}+\dfrac{p\, B_1}{n}+\dfrac{p\,L}{n} \,\,\, \geq \,\,\, -p\,\beta B_1-p\,R_H-p\,X +p\, B_1 +p\,L$
\\[0.1in]
$\equiv$ & 
$-\dfrac{\beta B_1p}{n}-\dfrac{p\,R_H}{n}-\dfrac{p\,X}{n}+\dfrac{p \,B_1}{n}+\dfrac{p\,L}{n}+p\,\beta B_1+p\,R_H+p\,X -p\, B_1 -p\,L \,\,\, \geq \,\,\, \dfrac{2RI}{n}-\dfrac{2\beta RI}{n}$
\\[0.1in]
$\equiv$ & 
$p \big(-\beta B_1-R_H-X+ B_1+L-n\beta B_1+nR_H+nX-n B_1-nL \big) \,\,\, \geq \,\,\, 2Ir\, \big(1-\beta \big)$
\\[0.1in]
$\equiv$ & 
$p \big(\beta B_1(n-1)+R_H(n-1)+X(n-1)- B_1(n-1)-L(n-1) \big) \,\,\, \geq \,\,\, 2Ir \, \big(1-\beta \big)$
\\[0.1in]
$\equiv$ & 
$p \,\,\, \geq \,\,\, \dfrac{2Ir\,\big(1-\beta\big)}{\big(n-1\big)\,\big(\beta  B_1+R_H+X- B_1-L\big)}
\,\,\,\equiv\,\,\,
p \,\,\, \geq \,\,\, \dfrac{2Ir\, \big(1-\beta \big)}{\big(n-1\big) \, \big(R_H+X-L- B_1(1-\beta ) \, \big)}$
\end{longtable}

\noindent
and thus $p^{\mathrm{soc}} = \dfrac{2Ir\,\big(1-\beta\big)}{\big(n-1\big) \,\big(R_H+X-L- B_1(1-\beta ) \, \big)}$.
\qed
\end{proof}

\begin{corollary}
If each bank has only two counter-parties then $r=1$ and thus 
$p^{\mathrm{soc}} = \dfrac{2I(1-\beta)}{\big(n-1\big)\,\big(R_H+X-L- B_1(1-\beta ) \,\big)}$.
\end{corollary}

\begin{lemma}[Probability of failure in risky autarky]\label{claim10}~\\
\[
p^{\mathrm{r.aut}} = (1-\beta)\dfrac{n}{n-1} \, X \, \dfrac{\displaystyle 1-\dfrac {2r!}{2^{2r}} \sum_{K=0}^r \dfrac{1}{K! \, (2r-K)!}}{\displaystyle R_H + X - L - B_1( 1- \beta)-\dfrac{(n\beta-1)X}{n-1} \left(1-\dfrac {2r!}{2^{2r}}  \sum_{K=0}^r \dfrac{1}{K! \, (2r-K)!} \right)}
\]
\end{lemma}

\noindent
({\em An informal explanation of Lemma~\ref{claim10} is as follows.
If a bank has enough equity to survive even if all of its counter-parties collapse then the bank may be forced to borrow less and this is not profitable. 
But, if the bank decides to deviate to risky autarky then it cannot pay back its debt if the real project delivers $R_H - 2ru$. 
The probability of failure if the bank chooses risky autarky is given by 
$p^{\mathrm{r.aut}} = (1-\beta)\frac{n}{n-1} \, X \, \frac{1-\frac {2r!}{2^{2r}} \sum_{K=0}^r \frac{1}{K! \, (2r-K)!}}{R_H + X - L - B_1( 1- \beta)-\frac{(n\beta-1)X}{n-1} \left(1-\frac {2r!}{2^{2r}}  \sum_{K=0}^r \frac{1}{K! \, (2r-K)!} \right)}$}).

\begin{proof}
\smartqed
For a bank to be in risky autarky, they borrow less such that it can roll over the debt at $t=1$ if its neighbors collapse but need not survive at $t=2$.
The relevant incentive constraint is: 
\begin{multline*}
R_H + \dfrac {2r!\,X}{2^{2\,r}} \sum_{K=0}^r \dfrac{1}{K! \, (2r-K)!}- {R^{\mathrm{r.aut}}_{\ensuremath{\imath},0}} D^{\mathrm{r.aut}} \geq B_1
\\
\equiv\,\,
R^{\mathrm{r.aut}}_{\ensuremath{\imath},0} D^{\mathrm{r.aut}} \leq R_H- B_1+ \dfrac {2r! \, X}{2^{2\,r}}  \sum_{K=0}^r \dfrac{1}{K! \, (2r-K)!}
\end{multline*}
The break-even condition for investors at $t=0$ is:
\begin{multline*}
D^{\mathrm{r.aut}} = \dfrac {p}{n} L +\left(1- \dfrac{p}{n} \right) R^{\mathrm{r.aut}}_{\ensuremath{\imath},0} D^{\mathrm{r.aut}}
\\
\equiv\,\,
D^{\mathrm{r.aut}} = \dfrac{p}{n} L +R_H- B_1 - \dfrac{p}{n} R_H + \dfrac{p}{n} B_1+ \left(1- \dfrac{p}{n} \right) \dfrac {2r!\,X}{2^{2r}}  \sum_{K=0}^r \dfrac{1}{K!\,(2r-K)!}
\end{multline*}
A bank decides not to deviate to risky autarky in a contagious system if and only if the following holds:

\begin{longtable}{ll}
& 
payoff in contagious system $\,\, \geq \,\,$ payoff in risky autarky
\\[0.1in]
$\equiv$ & 
$\beta \, (1-p) B_1 - \left(1-D^{\,\star}\right) \,\, \geq \,\, \beta \, (1-p) \left(R_H + X - R^{\mathrm{\,r.aut}}_{\ensuremath{\imath},0}\, D^{\mathrm{\,r.aut}} \right) 
          + \beta \, \dfrac{n-1}{n} p \, B_1 - \left(1 - D^{\mathrm{r.aut}}\right)$
\\[0.1in]
$\equiv$ & 
$\beta B_1- p \, \beta B_1-1+R_H+X-B_1-p \, R_H-p \, X+p \, B_1$
\\
& 
\hspace*{0.3in}
$\geq \,\,\big(\beta-p  \, \beta \big) \, \left( R_H + X-R_H +B_1 -  \dfrac {2r! \, X}{2^{2r}}  \sum_{K=0}^r \dfrac{1}{K! \, (2r-K)!}\right) + \beta \, \dfrac{n-1}{n} p \, B_1 - 1+\dfrac {p}{n}\,L$
\\
& 
\hspace*{0.5in}
$+\, R_H-B_1-\dfrac{p}{n}R_H+\dfrac{p}{n}B_1+\left(1-\dfrac{p}{n}\right)\dfrac {2r! \, X}{2^{2r}}  \sum_{K=0}^r \dfrac{1}{K! \, (2r-K)!}$
\\[0.1in]
$\equiv$ & 
$\dfrac{n}{n-1}(1-\beta)X \left( 1-\dfrac {2r!}{2^{2r}}  \sum_{K=0}^r \dfrac{1}{K! \, (2r-K)!} \right)$
\\
&
\hspace*{0.3in}
$\geq \,\, p \left[ R_H+X-L-B_1(1-\beta)-\dfrac{(n\beta-1)X}{n-1} \left(1-\dfrac {2r!}{2^{2r}}  \sum_{K=0}^r \dfrac{1}{K! \, (2r-K)!} \right) \, \right]$
\end{longtable}

\noindent
which implies 
\[
p^{\mathrm{r.aut}} = (1-\beta)\dfrac{n}{n-1} \, X \, \dfrac{\displaystyle 1-\dfrac {2r!}{2^{2r}} \sum_{K=0}^r \dfrac{1}{K! \, (2r-K)!}}{\displaystyle R_H + X - L - B_1( 1- \beta)-\dfrac{(n\beta-1)X}{n-1} \left(1-\dfrac {2r!}{2^{2r}}  \sum_{K=0}^r \dfrac{1}{K! \, (2r-K)!} \right)}
\]
\qed
\end{proof}

\begin{corollary}
If the number of neighbors is $2$ then $r=1$ and thus

\begin{longtable}{ll}
&
$p^{\mathrm{r.aut}} = (1-\beta) \dfrac{n}{n-1} \, X \, \dfrac{1-\dfrac{2}{4} \left( \dfrac{1}{2}+\dfrac{1}{1} \right) }{R_H + X - L - B_1 ( 1- \beta)-\dfrac{n\beta - 1}{n-1} \,X \, 
      \left( 1-\dfrac{2}{4} \left( \dfrac{1}{2}+\dfrac{1}{1} \right) \, \right)}$
\\ [0.1in]
$\equiv$ &
$p^{\mathrm{r.aut}} = (1-\beta) \dfrac{n}{n-1} \, X \, \dfrac{1-\dfrac{3}{4}}{R_H + X - L - B_1( 1- \beta)-\dfrac{n\beta - 1}{n-1} \, X \, \left(1-\frac{3}{4} \right)}$
\\ [0.1in]
$\equiv$ &
$p^{\mathrm{r.aut}} = (1-\beta) \dfrac{n}{n-1} \dfrac{\dfrac{X}{4}}{R_H + X - L - B_1( 1- \beta)- \left(\dfrac{n\beta - 1}{n-1} \right) \, \left(\dfrac{X}{4} \right)}$
\end{longtable}
\end{corollary}

\begin{lemma}[Probability of failure in safe autarky]\label{claim11}~\\
\[
p^{\mathrm{s.aut}} = (1-\beta )\, \left(\dfrac{n}{n-1} \right) \, \dfrac{2ru+X - B_1 }{  R_H+X- L - B_1+ \beta B_1 + \dfrac{1 - \beta}{n-1} \big(X+ 2 r u-B_1 \big) }
\]
\end{lemma}

\noindent
({\em An informal explanation of Lemma~\ref{claim11} is as follows. If a bank chooses to deviate to safe autarky, it will survive unless it is 
directly affected by low return $R_L$. It can pay back its debt even if the real project delivers $R_H - 2ru$.  
The probability of failure if the bank chooses safe autarky is given by 
$p^{\mathrm{s.aut}} = (1-\beta )\, \left(\frac{n}{n-1} \right) \, \frac{2ru+X - B_1 }{  R_H+X- L - B_1+ \beta B_1 + \frac{1 - \beta}{n-1} \big(X+ 2 r u-B_1 \big) }$}).

\begin{proof}
\smartqed
Suppose that a bank survives at $t=2$ even if all of its counter-party risks are realized to ensure that the payoff of the real project is $R_H - 2ru$. 
Since the bank needs enough equity to survive even if real project yields $R_H - 2ru$, we have 
\[
R_H-2ru-R^{\mathrm{s.aut}}_{\ensuremath{\imath},0} D^{\mathrm{s.aut}}\geq 0 
\,\,\equiv\,\,
R^{\mathrm{s.aut}}_{\ensuremath{\imath},0} D^{\mathrm{s.aut}}\leq R_H-2ru
\]
The break-even condition for investors at $t=0$ is:
\[
D^{\mathrm{s.aut}} = \left(1-\dfrac{p}{n}\right) \, R^{\mathrm{s.aut}}_{\ensuremath{\imath},0} D^{\mathrm{s.aut}}+\dfrac{p}{n}L 
\,\,\equiv\,\,
D^{\mathrm{s.aut}} = \left(1-\dfrac{p}{n}\right) \, \big(R_H-2ru \big) +\dfrac{p}{n}L
\]
Banks do not deviate to safe autarky from contagious system if and only if the following holds:

\begin{longtable}{ll}
& payoff in contagious system $\,\, \geq  \,\,$ payoff in safe autarky
\\[0.1in]
$\equiv$ & 
$\beta \big(1-p\big)B_1 - \left(1-D^{\,\star}\right) \,\, \geq \,\, \beta \left(1-\dfrac{p}{n} \right) \, \left(R_H+X-R^{\mathrm{s.aut}}_{\ensuremath{\imath},0} D^{\mathrm{s.aut}} \right)
              - \left(1-D^{\mathrm{s.aut}}\right)$
\\[0.1in]
$\equiv$ & 
$\beta B_1 - p \beta B_1 - 1 + R_H+X-B_1-p \, R_H-p \, X+p \, B_1+p \, L$ 
\\[0.1in]
& \hspace*{0.5in} $\geq \,\,
\left(\beta - \dfrac{p\beta}{n} \right) \, \big(R_H+X-R_H+2ru \big)-1+R_H-2ru-\dfrac{p}{n}R_H+\dfrac{p}{n}\,\big(2ru\big)+\dfrac{p}{n}\,L$
\\[0.1in]
$\equiv$ & 
$\beta B_1 - p \beta B_1+X-B_1-pR_H-p\,X+p\,B_1+p\,L$
\\
& \hspace*{0.5in} $\geq \,\,
\beta R_H+\beta X-\beta R_H+2 \beta ru - \dfrac{p\beta R_H}{n}-\dfrac{p\beta X}{n}+\dfrac{p\beta R_H}{n}-\dfrac{2p\beta ru}{n}-2ru-\dfrac{p}{n}R_H+\dfrac{p}{n} \big(2ru\big) +\dfrac{p}{n}\,L$
\\[0.1in]
$\equiv$ & 
$\beta B_1 +X - B_1 - \beta X-2 \beta r u + 2 r u  \,\, \geq \,\, \dfrac{p}{n} \, \Bigg( n \beta B_1 + n R_H+nX-n B_1 - n L - \beta X - 2 r \beta u - R_H + 2 r u +L \Bigg)$
\\[0.1in]
$\equiv$ & 
$(1-\beta )X -(1- \beta ) B_1 +2 r u(1- \beta )$ 
\\
& \hspace {0.2in} $\geq\,\,
\dfrac{p}{n} \, \Bigg( n \beta B_1 +  R_H(n-1)+nX-n B_1 - L(n-1) - \beta X + 2 r u(1- \beta )+X-X+B_1-B_1+\beta B_1 - \beta B_1 \Bigg)$
\\[0.1in]
$\equiv$ & 
$(1-\beta )\,n \,\big(2ru+X - B_1 \big)$
\\
& \hspace {0.2in} $\geq \,\,
p \, \Bigg( \beta B_1(n-1) +  R_H(n-1)+X(n-1)- B_1(n-1) - L(n-1) + X(1 - \beta) + 2 r u(1- \beta )-B_1(1-  \beta) \, \Bigg)$
\\[0.1in]
$\equiv$ & 
$(1-\beta )\dfrac{n}{n-1}\,(2Ru+X - B_1 ) \,\,\geq \,\,p \, \Bigg( \beta B_1 +  R_H+X- B_1 - L + \dfrac{1 - \beta}{n-1}(X+ 2 r u-B_1) \, \Bigg)$
\\[0.1in]
$\equiv$ & 
$p \,\,\leq \,\,(1-\beta )\,\left(\dfrac{n}{n-1} \right) \, \dfrac{2ru+X - B_1 }{  R_H+X- L - B_1+ \beta B_1 + \dfrac{1 - \beta}{n-1} \big(X+ 2 r u-B_1 \big)}$
\end{longtable}

\noindent
This implies 
$p^{\mathrm{s.aut}} = (1-\beta ) \, \left(\dfrac{n}{n-1} \right) \, \dfrac{2ru+X - B_1 }{  R_H+X- L - B_1+ \beta B_1 + \dfrac{1 - \beta}{n-1} \big(X+ 2 r u-B_1 \big)}$.
\qed
\end{proof}

\begin{corollary}
If the number of neighbors is $2$ then $r=1$ and thus 
\[
p^{\mathrm{s.aut}} = \big(1-\beta \big) \, \left(\dfrac{n}{n-1} \right) \, \dfrac{2u+X - B_1 }{  R_H+X- L - B_1+ \beta B_1 + \dfrac{1 - \beta}{n-1} \big(X+ 2 u-B_1 \big)}
\]
\end{corollary}

\begin{lemma}\label{claim6}
In equilibrium banks hedge all of its counter-party risks i.e.  banks endogenously enter into OTC contracts.
\end{lemma}

\begin{proof}
\smartqed
The banks hedge all of its counter-party risks if and only if
\[
\begin{array}{ll}
& \text{payoff from hedging in contagious equilibrium } > \text{ payoff from not hedging}
\\ [0.1in]
\equiv &
\beta\,\big(1-p\big) B_1 
\\
&
\hspace*{0.3in}
> \, \beta\,\Big(1-p\Big) \dfrac{2r!}{2^{2r}} \, \left(\dfrac{B_1 + 2ru}{0!\, (2r-0)!}+\dfrac{ B_1 + (2r-2)u}{1! \, (2r-1)!}+\dfrac{ B_1 + (2r-4)u}{2! \, (2r-2)!}+\dots+\dfrac{ B_1 + (2r - 2r)u}{\dfrac{2r}{2}! \left(2r-\dfrac{2r}{2} \right)!} \right)
\\ [0.1in]
\equiv &
B_1 > \dfrac{2r!}{2^{2r}}\left(\dfrac{ B_1 + 2ru}{0! \, (2r-0)!}+\dfrac{ B_1 + (2r - 2)u}{1! \, (2r-1)!}+\dfrac{ B_1 + (2r - 4)u}{2! \, (2r-2)!}+\dots+\dfrac{ B_1 + (2r - 2r)u}{\dfrac{2r}{2}! \, \left(2r-\dfrac{2r}{2} \right)!} \right)
\\ [0.1in]
\equiv &
B_1>\dfrac{\displaystyle 2 \, u \, r\sum_{K=0}^r\dfrac{2r-2K}{K! \, (2r-K)!}}{\displaystyle 2^{2r}-\sum_{K=0}^r \dfrac{2r!}{K!\, (2r-K)!}}
\end{array}
\]
which satisfies inequality~\eqref{eq4}, and thus banks hedge all of its counter-party risks.

Let the returns from the successful project be $\displaystyle R_H \, + \!\!\! \sum_{k\,=\,|\,i-1\,|}^{|\,i-r\,|} \!\!\! \varepsilon_k \,\,  - \!\!\!\! \sum_{k\,=\,|\,i-(r+1)\,|}^{|\,i-2r\,|} \!\!\! \!\!\! \!\!\! \varepsilon_k$. 
Assume that the bank goes bankrupt at $t=2$ if the counter-party realization of its unhedged risk is 
$-u$. This is true if the bank cannot repay its debt at $t=2$, \IE, 
\[
R_H-u < R^{\,\star}_{\ensuremath{\imath},0}\,D^{\,\star}
\,\equiv\,
R_H-u < R_H- B_1+X
\,\equiv\,
u > B_1-X
\]
if a bank fails when it loses $-u$ on its unhedged counter-party exposure, it will fail when the loss is greater than $-u$ on its unhedged counter-party exposure.
\qed
\end{proof}

\begin{corollary}
If number of counter-party neighbors is $2$ ({\em\IE}, $r=1$), then 
\\
$\displaystyle B_1 > \dfrac{2u \left(\dfrac{2-0}{0!\,(2-0)!}+\dfrac{2-2}{1!\,(2-1)!} \right)}{2^2-\left(\dfrac{2!}{0!\,(2-0)!}+\dfrac{2!}{1!\,(2-1)!}\right)}
\,\equiv\,
B_1 >2u$
\hspace*{0.2in}and\hspace*{0.2in}
$u > B_1-X$.
\end{corollary}

\begin{lemma}\label{claim13}
When the number of neighboring counter-parties is less than $\nicefrac{n}{2}$, a bank chooses to shirk if one or more of its counter-parties default 
by leaving the bank unhedged resulting in its debt not being rolled over at $t=1$.
\end{lemma}

\begin{proof}
\smartqed
If a bank borrows $D^{\max}(R_{l,0})$ at $t = 0$ and if the bank has a low expected realization of $R_L$,
then the debt financing is not rolled over at $t=1$. Since $R_L<L$, the creditors will want to terminate the project. The bank goes bankrupt if its debt financing is not rolled over at $t=1$.

If $2r$, the number of counter-party neighbors, is less than $\nicefrac{n}{2}$, then probability of failure due to counter-party risk is 
less than $\frac{\nicefrac{n}{2}}{n}$, \IE, the probability of counter-party risk is less than $\nicefrac{1}{2}$. Since we assume that banks will consider a 
counter party risk probability of at least $\nicefrac{1}{2}$ to insure against counter-party risk, banks do not insure against counter-party risk using default insurance 
when $2r<\nicefrac{n}{2}$. When $2r\geq\nicefrac{n}{2}$, the probability of counter-party risk becomes at least $\nicefrac{1}{2}$ and hence banks will hedge the counter-party risk. 
The counter-party insurance payoff happens with probability $\frac{2rp}{n}$ in case of private perspective and with probability $\frac{(n-1)p}{n}$ in case of social perspective. 
Thus, as $2r$ increases to $n-1$, the counter-party insurance payoff probability in private perspective becomes the same as that in social perspective. 

When $2r \geq \nicefrac{n}{2} $, the individual banks will hedge the counter-party risk by taking counter-party insurance.
When $ 2r < \nicefrac{n}{2} $, banks will not hedge the counter-party risk and hence failure of a counter-party will lead to the violation of its incentive constraint, thus the bank shirks and the project delivers $R_L$. 
Let $D_1$ be the amount of debt to be rolled over at $t=1$. The investors will demand higher interest rate $R_{\ensuremath{\imath},1}$ in order to break even. 
Let $P_s$ be the probability of a bank that do not default, $P_f$ be the probability of a bank that defaults, and $n_d$ be the number of neighbors of any bank that default.
Thus, $P_s=\dfrac{2r-n_d}{2r}$ and $P_f=\dfrac{n_d}{2r}$. By the break even condition of investors, we get
\[
P_s R_{\ensuremath{\imath},1} D_1 + P_f R_L = D_1
\,\,\equiv\,\,
R_{\ensuremath{\imath},1} D_1  = \dfrac {D_1 -  P_f R_L}{P_s}
\]
The incentive constraint is
\[
\begin{array}{ll}
& \beta (R_H - R_{\ensuremath{\imath},1} D_1   + X ) \geq \beta B_1
\\[0.10in]
\equiv &
\beta \left(R_H - \dfrac {D_1 -  P_f R_L}{P_s} + X \right) \geq \beta B_1
\\[0.10in]
\equiv &
P_s R_H - D_1 +  P_f R_L +P_s X \geq B_1 P_s
\\[0.10in]
\equiv &
P_s R_H - D_1 +  P_f R_L +P_s X \geq B_1 P_s
\\[0.10in]
\equiv &
P_s \big(R_H+X-B_1\big) + P_f R_L \geq D_1 
\end{array}
\]
If the bank had originally borrowed $D_0$, the amount it has to roll over at $t=1$ is $R_{\ensuremath{\imath},0} D_0 = R_H + X -B_1$, 
but the amount that is actually getting rolled over is only $P_s (R_H+X-B_1)+  P_f R_L$. 
\qed
\end{proof}

\begin{corollary}
If number of counter-party neighbors is $2$ ({\em\IE}, $2r=2$) and number of banks defaulted is $1$ ({\em\IE}, $n_d=1$), then 
$P_s=\frac{2r-n_d}{2r}=\nicefrac{1}{2}$, $P_f=\frac{n_d}{2r}=\nicefrac{1}{2}$, and thus 
\[
P_s \big(R_H+X-B_1\big) + P_f R_L \geq D_1 
\,\,\,\equiv \,\,\,
\dfrac {1}{2} \big(R_H+X-B_1\big) + \dfrac {1}{2} R_L \geq D_1 
\]
\end{corollary}


\section*{Appendix 1: Remaining Proofs}
\addcontentsline{toc}{section}{Appendix 1: Remaining Proofs}

\section*{(I) Proof of $D^{\max}=\frac{R_H- B_1+X}{R_{\ensuremath{\imath},0}}<1$}

The incentive constraint of a bank surviving at $t=1$ and holding no risk can be written as follows:
\begin{gather*}
\text{payoff if bank exerts effort } \geq \text{ payoff if bank does not exert effort}
\\
\,\,\,\equiv\,\,\,
\beta [ R_H- R_{\ensuremath{\imath},0} D+ X ] \geq \beta B_1
\,\,\,\equiv\,\,\,
\dfrac{R_H- B_1+X}{R_{\ensuremath{\imath},0}} \geq  D
\hspace*{-0.2in}
\underbrace{\equiv\,\,\, D^{\max}=\dfrac{R_H- B_1+X}{R_{\ensuremath{\imath},0}} < 1} 
_{\mbox{
\begin{tabular}{c}
since $R_{\ensuremath{\imath},0}\geq 1$
\\
and, by inequality~\eqref{eq3}, $R_H +X-B_1\leq 1$
\end{tabular}
}
}
\end{gather*}
%


\section*{(II) Proof of $s^{\mathrm{safe}} = \frac{1-\beta}{n}- \beta\frac{p}{n}$}

Assume all counter parties are insured. The insurance company has to hold a capital of $2RI$. If all the banks insure against the failure of their counter-parties, 
the insurance fund breaks even if the price per unit of the default insurance $s$ is determined by the break even condition. Using the superscript 
$\scriptstyle\mathrm{safe}$ to denote the insured system, we get
\[
\begin{array}{l}
\text{$\Big($expected amount that remains in the insurance fund at $t=2\Big)$} 
\\
\hspace*{0.4in}
= \text{ $\Big($amount that insurance fund has to set aside at $t=0\Big)$}
\\ 
\equiv \,\,\,
\beta (2rI - 2prI) = 2rI - 2nrIs
\,\,\,\equiv \,\,\,
\beta \big(1-p\big) = 1-ns
\,\,\,\equiv \,\,\,
s^{\mathrm{safe}} = \dfrac{1-\beta}{n}- \beta\dfrac{p}{n} 
\end{array}
\]
%


\section*{(III) Proof of $D^{\mathrm{safe}} = \left(1-\frac{p}{n}\right) (R_H- B_1+X) + \frac{p}{n}L$}

Since the investors break even, we have $\left(1-\dfrac{p}{n}\right)R^{\mathrm{safe}}D^{\mathrm{safe}}+\dfrac{p}{n}L = D^{\mathrm{safe}}$.
Combining with the maximum amount that can be borrowed, we get
\[
D^{\max} =\dfrac{R_H- B_1+X}{R_{\ensuremath{\imath},0}}
\,\,\, \Longrightarrow  \,\,\,
R^{\mathrm{safe}}D^{\mathrm{safe}}=R_H- B_1+X
\]
which proves our claim.


\section*{(IV) Proof of $D^{\star} = (1-p)(R_H- B_1+X) +p\,L$}

If all the banks decide not to buy insurance and hold the minimum amount of equity to roll over debt, the system is contagious. 
The investors will always anticipate the equilibrium and break even, implying 
\[
(1-p)\,R^\star_{\ensuremath{\imath},0}\,D^\star +p\,L = D^\star
\,\,\,\underbrace{\equiv \,\,\, D^{\star} = (1-p)(R_H- B_1+X) +p\,L}_{\mbox{since by {\bf (I)} $\,R^\star_{\ensuremath{\imath},0}\,D^\star = R_H- B_1+X$}}
\]
%


\section*{(V) Proof of $R^\star_{\ensuremath{\imath},0} = \frac{1}{1-p \, \left(1-\frac{L}{R_H- B_1+X} \right) }$}

\begin{multline*}
\underbrace{R^\star_{\ensuremath{\imath},0}\,D^\star = R_H- B_1+X}_{\mbox{by {\bf (I)}}}
\,\,\,
\equiv 
\,\,\,
\underbrace{
R^\star_{\ensuremath{\imath},0}= \dfrac{R_H- B_1+X}{D^\star}
\,\,\,
\equiv \,\,\, R^\star_{\ensuremath{\imath},0}= \dfrac{R_H- B_1+X}{(1-p)(R_H- B_1+X)+pL}}_{\mbox{since, by {\bf (IV)}, 
$D^{\star} = (1-p)(R_H- B_1+X) +p\,L$
}}
\\
\equiv 
\,\,\,
R^\star_{\ensuremath{\imath},0} = \frac{1}{1-p \, \left(1-\frac{L}{R_H- B_1+X} \right)} 
\end{multline*}`
%


\section*{(VI) Proof of $p^{\mathrm{term}} = (1 - \beta)\left(\frac{n}{n-1}\right)\left(\frac {B_0}{R_H + X - L - B_1(1-\beta B_1  ) +\frac{R_L+ \beta B_1 - L}{n-1}}\right)$}

If a bank borrows long term, the banker can borrow up to the point to make him/her indifferent between long-term borrowing and shirking in both periods to collect $B_0+ B_1$.
In long-term lending, investors cannot liquidate the bank under any circumstances. Since shirking in only the second period is clearly suboptimal, the banker shirks in both periods and 
the expected payoff as of $t=0$ is $B_0+B_1$, and thus the payoff in the good state must be at least $B_0+B_1$. This implies:
\begin{gather*}
\left(1-\dfrac{p}{n} \right) \left(R_H-R^{\,\mathrm{term}} D^{\,\mathrm{term}}+X \right)+\dfrac{p}{n} B_1 \geq B_0+ B_1
\end{gather*}
Considering the worst case scenario
\begin{multline}
\left(1-\dfrac{p}{n} \right) \left(R_H-R^{\,\mathrm{term}} D^{\,\mathrm{term}}+X \right)+\dfrac{p}{n} B_1 =  B_0+ B_1
\\
\equiv \,\,\,
R^{\,\mathrm{term}} D^{\,\mathrm{term}} = R_H+X-\dfrac{ B_0+ B_1-\dfrac{p}{n} B_1}{1-\dfrac{p}{n}}
\label{teq1}
\end{multline}
Since the lenders break even in expectation at $t=0$, we have 
\begin{multline*}
D^{\,\mathrm{term}} = \left(1-\dfrac{p}{n}\right) R^{\,\mathrm{term}} D^{\,\mathrm{term}}+\dfrac{p}{n} R_L
\,\,\, \Longrightarrow \,\,\,
\underbrace{D^{\,\mathrm{term}} =  \left(1-\dfrac{p}{n}\right)(X+R_H)-B_0-  B_1+\dfrac{p}{n}  B_1 +\dfrac{p}{n} R_L}_{\mbox{using~\eqref{teq1}}}
\\
\equiv \,\,\,
D^{\,\mathrm{term}} =  X+R_H-\dfrac{p}{n}X-\dfrac{p}{n}R_H-B_0-  B_1+\dfrac{p}{n}  B_1 +\dfrac{p}{n} R_L
\end{multline*}
Banks will not decide to deviate from contagious equilibrium provided the following holds:

\begin{longtable}{ll}
& payoff in contagious system $\,\geq\,$ payoff in long-term borrowing
\\
[0.1in]
$\displaystyle \equiv$ & 
$\displaystyle \beta (1-p).  B_1 - (1-D^{*}) \geq \beta (B_0+ B_1) - (1- D^{\,\mathrm{term}})$
\\
[0.1in]
$\displaystyle \equiv$ & 
$\displaystyle \beta B_1 -p\beta B_1 -1 +R_H+X-B_1 -pR_H-pX+pB_1+pL$
\\
& 
\hspace*{1in} 
$\displaystyle \geq \beta B_0 +\beta B_1-1+X+R_H-\dfrac{p}{n} X-\dfrac{p}{n} R_H - B_0 - B_1+\dfrac{p}{n} B_1 +\dfrac{p}{n} R_L$
\\
[0.1in]
$\displaystyle \equiv$ & 
$\displaystyle B_0 - \beta B_0 \geq \dfrac{p}{n} \Big(n\beta B_1 +nR_H+nX - nB_1 - nL- X- R_H +B_1+ R_L + L - L+ \beta B_1 - \beta B_1 \Big)$
\\
[0.1in]
$\displaystyle \equiv$ & 
$\displaystyle n B_0 (1 - \beta) \geq p(\beta B_1 (n-1) + R_H (n-1) + X ( n-1) - B_1 (n -1 ) - L ( n - 1 ) + R_L+ \beta B_1 - L)$
\\
[0.1in]
$\displaystyle \equiv$ & 
$\displaystyle \dfrac {n}{n-1} B_0 (1 - \beta) \geq p \left(\beta B_1  + R_H + X  - B_1  - L  +\dfrac{ R_L+ \beta B_1 - L}{n-1} \right)$
\\
[0.1in]
$\displaystyle \equiv$ & 
$\displaystyle p \leq (1 - \beta) \, \left( \dfrac {n}{n-1} \right) \, \left( \dfrac {B_0}{ R_H + X - L - B_1(1-\beta B_1  ) +\dfrac{ R_L+ \beta B_1 - L}{n-1}} \right)$
\\
[0.1in]
$\displaystyle \Longrightarrow$ & 
$\displaystyle p^{\,\mathrm{term}} = (1 - \beta) \, \left( \dfrac {n}{n-1} \right) \, \left( \dfrac {B_0}{ R_H + X - L - B_1(1-\beta B_1  ) +\dfrac{ R_L+ \beta B_1 - L}{n-1}} \right)$
\end{longtable}
%

\section*{(VII) Proof of $p^{\,\mathrm{f.aut}} = \frac{(1- \beta) (R_H + X - B_1 - R_L)}{R_H + X - L -B_1(1 - \beta) - \frac{\beta}{n}(R_H - R_L)}$}

Since a bank survives at $t=2$ even if its real project only delivers $R_L$, we have 
\[
\begin{array}{ll}
& \text{payoff in contagious system } \geq \text{ payoff in full autarky}
\\[0.1in]

\equiv & 
\beta (1-p) B_1 - (1- D^\star) \, \geq \, \beta \left(1- \dfrac{p}{n} \right)(R_H + X - R_L )+ \beta \dfrac{p}{n} X - \big(1 - R_L\big)
\\[0.1in]
\equiv & 
\beta B_1 - p \beta B_1 - 1+R_H+X-B_1-pR_H-pX+pB_1+pL \, \geq \, 
\\[0.1in]
 & 
\beta R_H + \beta X - \beta R_L -\dfrac{p}{n} \beta R_H -\dfrac{p}{n} \beta X +\dfrac{p}{n} \beta R_L + \beta \dfrac{p}{n} X - \big(1 - R_L\big)
\\[0.1in]
\equiv & 
\beta B_1 +R_H+X-B_1- \beta X +\beta R_L- R_L- \beta R_H \, \geq \, p\left(\beta B_1 + R_H + X -  B_1 - L - \dfrac{\beta R_H}{n}+\dfrac{\beta R_L}{n} \right)
\\[0.1in]
\equiv & 
R_H(1 - \beta )+X(1 - \beta )-B_1(1 - \beta) - R_L(1 - \beta) \, \geq \, p \left( R_H + X - L -  B_1(1 - \beta ) - \dfrac{\beta}{n} \left(R_H - R_L\right) \right)
\\[0.1in]
\equiv & 
p \leq \dfrac{(1- \beta) (R_H + X - B_1 - R_L)}{R_H + X - L -B_1(1 - \beta) - \dfrac{\beta}{n}(R_H - R_L)}
\,\,\, \equiv \,\,\,
p^{\,\mathrm{f.aut}} = \dfrac {(1- \beta) (R_H + X - B_1 - R_L)}{R_H + X - L -B_1(1 - \beta) - \dfrac{\beta}{n}(R_H - R_L)}
\end{array}
\]
%

\section*{(VIII) Proof of $I= ru+X-B_1 > 0$}

The amount of debt reduction needed at time $t=1$ to stabilize a bank that lost at least $r$ counter-parties is $I=ru+X-B_1 > 0$. 
If less than $r$ counter-party neighbors fail, then the fraction of failed counter-parties of the bank is less than $\nicefrac{1}{2}$, hence bank can survive this loss.
If a bank looses at least $r$ hedges at $t=1$, it can only roll over its debt if its incentive constraint is not violated. 
This can be insured by injecting enough equity into the bank to make sure it can pay back its debt even if it loses $ru$ at $t=2$ due to the lost hedge.
Thus, we have 
\[
R_H - ru -  R{}^{\max}_{\ensuremath{\imath},1} D^{\max}+I \geq 0
\,\equiv\, R^{\max}_{\ensuremath{\imath},1} D^{\max} = R_H+X-B_1
\,\Longrightarrow\, I = ru+X-B_1
\]

\section*{Appendix 2: Glossary of financial terminologies}
\addcontentsline{toc}{section}{Appendix 2: Glossary of financial terminologies}

\begin{description}
\item[\bf Risk:] 
Risk is a chance that an investment's actual return will be less than expected.

\item[\bf Asset:] 
An asset is anything that is owned by an individual or business that has a monetary value.

\item[\bf Counter-party Risk:] 
Risk that one party in an agreement defaults on its obligation to repay or return securities.

\item[\bf Hedging:]
A strategy to reduce the risk by making a transaction in one market to protect against the loss in another market.

\item[\bf Over The Counter (OTC):]
A market that is conducted between dealers by telephone and computer and {\em not} on a listed exchange. 
OTC stocks tend to be of those companies that do not meet the listing requirements of an exchange, although some companies that do meet the listing requirements 
choose to remain as OTC stocks. The deals and instruments are generally not standardized and there is {\em no} public record of the price associated with any transaction.

\item[\bf Equity:] 
The amount that shareholders own, in the form of common or preferred stock, in a publicly quoted company. 
Equity is the risk-bearing part of the company's capital and contrasts with debt capital which is usually secured in some way and which has 
priority over shareholders if the company becomes insolvent and its assets are distributed.

\item[\bf Liquidity:]
The degree to which an asset or security can be bought or sold in the market without affecting the asset's price, \IE, the ability to convert an asset to cash quickly.

\item[\bf Credit Default Swap (CDS):] 
A specific type of counter-party agreement that allows the transfer of third party credit risk from one party to the other. 
One party in the swap is a lender and faces a credit risk from a third party, and the counter-party in the swap agrees to insure this risk in exchange for {\em regular periodic} payments.
\end{description}

\section*{Appendix 3: List of notations and variables}
\addcontentsline{toc}{section}{Appendix 3: List of notations and variables}

\renewcommand{\arraystretch}{1.5}
\begin{longtable}{r c p{4.2in}}
${n}$ & : & total number of banks in the network.
\\
${t}$ & : & the time variable. Three time periods $t=0,1,2$ are considered (initial, interim and final time period, respectively).
\\
${\beta}$ & : & the discount factor (we assume that $\beta<1$).
\\
${R}$ & : & the borrowing rate. $R_0$ and $R_1$ are the borrowing rates at $t=0$ and at $t=1$, respectively.
\\
${\varepsilon}$ & : & an independent random variable, realized at $t=2$, that takes a value of $+u$ or $-u$, each with probability $\nicefrac{1}{2}$. 
\\
${L}$ & : & the return if the investments are liquidated early at $t=1$ $L(<R_H)$.
\\
${X}$ & : & Non-pledgable payoff.
\\
${r}$ & : & each bank is assumed to have $2r$ counter-party neighbors.
\\
${D}$ & : & 
investors share of investment in the bank at $t=0$; $1-D$ is the bankers share of investment in the bank at $t=0$ ($D,1-D \geq 1$).
\\
${e}$ & : & the unobservable effort choice made by the bank ($e \in\{0,1\}$).
\\
${B_i}$ & : & Bankers private benefit at the time period $i$.
\\
${p}$ & : & the probability of bad state.
\\
${p^{\mathrm{soc}}}$ & : & if $p<p^{\mathrm{soc}}$, then irrespective of the number of counter-party neighbors there is 
no need for any counter-party insurance even from a social perspective. 
\\
${p^{\mathrm{ind}}}$ & : & 
if $p>p^{\mathrm{ind}}$ then the banks will not buy any counter-party insurance as the private benefits of insuring exceed the cost. 
\\
${p^{\mathrm{term}}}$ & : & 
if $p<p^{\mathrm{term}}$ then the banks will continue to prefer short-term debt.
\\
${p^{\mathrm{aut}}}$ & : & 
if $p<p^{\mathrm{aut}}$ then no bank will have an incentive to hold more equity and borrow less. 
$p^{\mathrm{aut}}$ is equal to $\min \left\{ p^{\mathrm{s.aut}},p^{\mathrm{r.aut}},p^{\mathrm{f.aut}}\right\}$
where the superscripts $\scriptstyle\mathrm{s.aut}$, $\scriptstyle\mathrm{r.aut}$ and $\scriptstyle\mathrm{f.aut}$ stand for 
safe autarky, risky autarky and full autarky, respectively.
\\
${D^{\max}}$ & : & The maximum amount of borrowing at $t=0$ that can be rolled over at $t=1$.
\\
${s^{\mathrm{safe}}}$ & : & The price per unit of default insurance in a stable system where all banks buy counter-party insurance.
\\
${D^{\mathrm{safe}}}$ & : & 
The amount borrowed by the bank for unit investment in the project in a stable system where all banks buy counter-party insurance.
\\
${D^{\,\star}}$ & : & 
The amount borrowed by the bank for unit investment in the project in a contagious system where banks decide not to buy insurance.
\end{longtable}
\renewcommand{\arraystretch}{1}

\end{document}